\documentclass[12pt,transaction, draftclsnofoot,onecolumn]{IEEEtran}
\usepackage{amsfonts}
\usepackage{amssymb}
\usepackage[draft]{hyperref}
\usepackage{graphicx}
\usepackage{subfigure}
\usepackage{enumerate}
\usepackage{amsmath}
\usepackage{color}
\usepackage{amsthm}
\usepackage{amsmath}
\usepackage{algorithm}
\usepackage{stfloats}
\usepackage{algpseudocode}

\newcommand{\bp}{\begin{proof} \small }
\newcommand{\ep}{\end{proof} \normalsize}
\newcommand{\epx}{\end{proof} \small}
\newcommand{\bpa}{\begin{proofappx} \footnotesize }
\newcommand{\epa}{\end{proofappx} \small }
\newtheorem{theorem}{Theorem}

\newtheorem{lemma}{Lemma}

\newtheorem*{theorem*}{Theorem}
\newtheorem*{proposition*}{Proposition}
\newtheorem*{corollary*}{Corollary}
\newtheorem*{lemma*}{Lemma}
\newtheorem*{assumption*}{Assumption}
\newtheorem*{definition*}{Definition}
\newtheorem*{claim*}{Claim}

\newcommand{\bm}[1]{\mbox{\boldmath $#1$}}

\newcommand{\be}{\begin{equation}}
\newcommand{\ee}{\end{equation}}
\newcommand{\bs}{\begin{subequations}}
\newcommand{\es}{\end{subequations}}
\newcommand{\bq}{\begin{eqnarray}}
\newcommand{\eq}{\end{eqnarray}}
\newcommand{\bqn}{\begin{eqnarray*}}
\newcommand{\eqn}{\end{eqnarray*}}

\newcommand{\ba}{\left[ \begin{array}}
\newcommand{\ea}{\\ \end{array} \right]}
\newcommand{\ben}{\begin{enumerate}}
\newcommand{\een}{\end{enumerate}}

\def\a{{\boldsymbol{a}}}
\def\b{{\boldsymbol{b}}}

\def\d{{\boldsymbol{d}}}

\def\real{{\mathchoice%
{\hbox{\rm\setbox1=\hbox{I}\copy1\kern-.45\wd1 R}}
{\hbox{\rm\setbox1=\hbox{I}\copy1\kern-.45\wd1 R}}
{\hbox{\scriptsize\rm\setbox1=\hbox{I}\copy1\kern-.45\wd1 R}}
{\hbox{\scriptsize\rm\setbox1=\hbox{I}\copy1\kern-.45\wd1 R}}}}

\def\Zint{{\mathchoice{\setbox1=\hbox{\sf Z}\copy1\kern-.75\wd1\box1}
{\setbox1=\hbox{\sf Z}\copy1\kern-.75\wd1\box1}
{\setbox1=\hbox{\scriptsize\sf Z}\copy1\kern-.75\wd1\box1}
{\setbox1=\hbox{\scriptsize\sf Z}\copy1\kern-.75\wd1\box1}}}
\newcommand{\complex}{ \hbox{\rm C\kern-0.45em\rule[.07em]{.02em}{.58em}%
\kern 0.43em}}

\begin{document}
%
\title{Joint Service Caching and Task Offloading for Mobile Edge Computing in Dense Networks}

\author{\IEEEauthorblockN{Jie Xu$^*$, Lixing Chen$^*$, Pan Zhou$^\dagger$\\}
\IEEEauthorblockA{$^*$Department of Electrical and Computer Engineering, University of Miami, USA\\
$^\dagger$School of EIC, Huazhong University of Science and Technology, China}
}

\maketitle

\begin{abstract}
Mobile Edge Computing (MEC) pushes computing functionalities away from the centralized cloud to the network edge, thereby meeting the latency requirements of many emerging mobile applications and saving backhaul network bandwidth. Although many existing works have studied computation offloading policies, service caching is an equally, if not more important, design topic of MEC, yet receives much less attention. Service caching refers to caching application services and their related databases/libraries in the edge server (e.g. MEC-enabled BS), thereby enabling corresponding computation tasks to be executed. Because only a small number of application services can be cached in resource-limited edge server at the same time, which services to cache has to be judiciously decided to maximize the edge computing performance. In this paper, we investigate the extremely compelling but much less studied problem of dynamic service caching in MEC-enabled dense cellular networks. We propose an efficient online algorithm, called OREO, which jointly optimizes dynamic service caching and task offloading to address a number of key challenges in MEC systems, including service heterogeneity, unknown system dynamics, spatial demand coupling and decentralized coordination. Our algorithm is developed based on Lyapunov optimization and Gibbs sampling, works online without requiring future information, and achieves provable close-to-optimal performance. Simulation results show that our algorithm can effectively reduce computation latency for end users while keeping energy consumption low.
\end{abstract}


%
\IEEEpeerreviewmaketitle

\section{Introduction}
Pervasive mobile computing and the Internet of Things are driving the development of many new compute-demanding and latency-sensitive applications, such as cognitive assistance, mobile gaming and virtual/augmented reality (VR/AR). Due to the often unpredictable network latency and expensive bandwidth, cloud computing becomes unable to meet the stringent requirements of latency-sensitive applications. The ever growing distributed data also renders it impractical to transport all the data over today's already-congested backbone networks to the remote cloud. To overcome these limitations, mobile edge computing (MEC) (a.k.a. fog computing) \cite{mao2017mobile,shi2016edge} has recently emerged as a new computing paradigm to push the frontier of computing applications, data and services away from centralized cloud computing infrastructures to the logical edge of a network, thereby enabling analytics and knowledge generation to occur closer to the data source.

\begin{figure}
  \centering
  \includegraphics[width=0.4\textwidth]{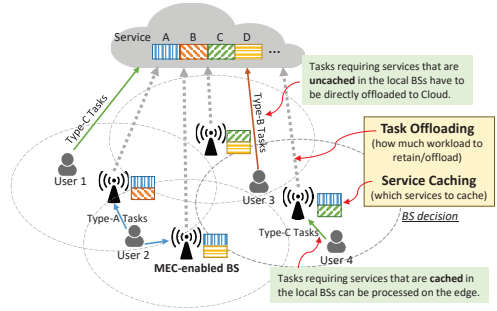}\\
  \caption{System illustration. Each BS can only cache a subset of services. Only user tasks requesting services cached in the BS can be executed by the BS. BSs have to jointly optimize service caching and task offloading. }\label{system}
  \vspace{-0.2 in}
\end{figure}

Considered as a major form of MEC, mobile base stations (BSs) endowed with cloud-like computing and storage capability, are able to serve end-users' computation requests as a substitute of the cloud \cite{mao2016dynamic}. Extra tasks exceeding the BS's computing capacity are further offloaded to the cloud, resulting in a hierarchical offloading structure among end-users, BSs and the cloud. While computation offloading has been the central theme of most recent works studying MEC, what is often ignored is the heterogeneity of mobile services and how these services are cached on BSs in the first place. Precisely, service caching (or service placement) refers to caching application services and their related databases/libraries in the edge server co-located with the BS, thereby enabling user tasks requiring these services to be executed. However, unlike the cloud which has huge and diverse resources, the limited computing and storage resources of BS allow only a small set of services to be cached at the same time.  As a result, which services are cached on the BS determines which tasks can be offloaded, thereby significantly affecting the edge computing performance. Figure \ref{system} provides a system illustration.

Optimal service caching faces many challenges. First, mobile services are heterogeneous in terms of not only required resources (e.g. online Matlab and AR services have different CPU and storage requirements \cite{mao2016dynamic}) but also popularity/demand among users. While the former is often fixed, the latter is changing both spatially and temporally. Therefore, service caching has to be adaptively updated in resource-limited BSs depending on the predicted service popularity.

Second, since the considered system operates in a highly stochastic environment with random demand arrivals, the long-term system performance is more relevant than the immediate short-term performance. However, the long-term resource constraint (e.g. energy consumption) couples the service caching decisions over time, and yet the decisions have to be made without foreseeing the future system dynamics.

Third, to accommodate the surging data demand of mobile users, the density of BSs in cellular networks has kept increasing since its birth to nowadays 4G networks, and is expected to reach about 40-50 BSs/km$^2$ in the next generation 5G networks \cite{ge20165g}. Dense cellular networks create a complex multi-cell environment where demand and resource are highly coupled in both spatial and temporal domains. Effective service caching and task offloading requires careful coordination among all BSs, and decentralized solutions are much favored in order to reduce complexity.

In this paper, we investigate the extremely compelling but much less studied problem of service caching in MEC-enabled cellular networks, and develop an efficient solution that jointly optimizes service caching and task offloading. The main contributions of this paper are summarized as follows.

(1) We formalize the joint service caching and task offloading problem in MEC-enabled dense cellular networks, for minimizing computation latency under a long-term energy consumption constraint. To our best knowledge, this is the first work that studies joint service caching and task offloading in multi-cell MEC systems.

(2) To solve this problem, we develop a novel online algorithm, called OREO (Online seRvice caching for mobile Edge cOmputing), to perform stochastic service caching in an online fashion without requiring future information. OREO is developed based on Lyapunov optimization, and we prove that it achieves close-to-optimal performance compared to the optimal algorithm with full future information, while bounding the potential violation of the energy consumption constraint.

(3) We develop a decentralized algorithm based on a variation of Gibbs sampling, which is a key subroutine of OREO, thereby enabling efficient decentralized coordination among the BSs. This makes our algorithm scalable to large networks.

(4) Extensive and practical simulations are carried out to evaluate the performance of the proposed algorithm.

The rest of this paper is organized as follows. Section II reviews related works. Section III presents the system model and formulates the problem. Section IV develops the OREO algorithm and analyzes its performance. Section V performs simulations, followed by the conclusion in Section VI.

\section{Related Work}
Computation offloading is the central theme of many prior studies in both mobile cloud computing (MCC) \cite{fernando2013mobile,buyya2009cloud} and mobile edge computing (MEC) \cite{mao2017mobile,shi2016edge}, which concerns what/when/how to offload users' workload from their devices to the edge servers or the cloud. Various works have studied different facets of this problem, considering e.g. stochastic task arrivals \cite{huang2012dynamic,liu2016delay}, energy efficiency \cite{xu2017online,sun2017emm,7997128}, collaboration \cite{chen2017socially,tanzil2016distributed,chen2017computation} etc. However, the implicit assumption is that edge servers can process whatever types of computation tasks that are offloaded from users without considering the availability of services in edge servers, which in fact is crucial in MEC due to the limited resources of edge servers.

Similar service caching/placement problems, known as virtual machine (VM) placement, have been investigated in conventional cloud computing systems. VM placement over multiple clouds is studied in \cite{tordsson2012cloud,li2009enacloud,gao2013multi}, where the goal is to reduce the deployment cost, maximize energy saving and improve user experience, given constraints on hardware configuration and load balancing. However, these works cannot be directly applied to design efficient service caching policies for MEC since mobile networks are much more complex and volatile, and optimization decisions are coupled both spatially and temporally. The most related work probably is \cite{yang2016cost}, which extends the idea to MCC systems and studied the joint optimization of service caching/placement over multiple cloudlets and load dispatching for end users' requests. There are several significant differences of our work. First, while the coverage areas are assumed to be non-overlapping for different cloudlets in \cite{yang2016cost}, BSs have overlapping coverage areas in our considered dense cellular network. Second, while only heuristic solutions are developed in \cite{yang2016cost}, we prove strong performance guarantee for our algorithm. Third, while the algorithm is centralized in \cite{yang2016cost}, our algorithm enables decentralized coordination among BSs. We also note that the term ``service placement'' was used by some other literature \cite{wang2017dynamic,zhang2013dynamic} in a different context. The concern there is to assign task instances to different clouds but there is no limitation on what types of services/applications that each cloud can run.

Service caching/placement is also related to content caching/placement in network edge devices \cite{shanmugam2013femtocaching}. For example, the authors of \cite{prabh2005energy} aim to find optimal locations to cache the data that minimize packet transmissions in wireless sensor nodes. The concept of FemtoCaching is introduced in \cite{shanmugam2013femtocaching} which studies content placement in small cell networks to minimize content access delay. The idea of using caching to support mobility has been investigated in \cite{wang2015dynamic}, where the goal is to reduce latency experienced by users moving between cells. Learning-based content caching policies are developed in \cite{muller2017context} for wireless networks with a priori unknown content popularity. While content caching mainly concerns with storage capacity constraints, service caching has to take into account both computing and storage constraints, and has to be jointly optimized with task offloading to maximize the overall system performance.

\section{System Model}
We divide time into discrete time slots each of which has a duration that matches the timescale at which service caching decisions can be updated. Although our model is not perfect, we believe that it is a reasonable and valuable first step towards studying dynamic service caching and task offloading in MEC systems. Future improvement directions are briefly discussed in the conclusion section.

\subsection{Network and Services}
We consider a mobile network of $N$ base stations (BSs), indexed by $\mathcal{N}$. Each BS $n \in \mathcal{N}$ is endowed with edge computing functionalities and hence can provide computing services to end users in its radio range. The network is divided into $M$ disjoint small regions, indexed by $\mathcal{M}$. User in each region $m$ can reach a set of BSs in the radio range, denoted by $\mathcal{B}_m \subseteq \mathcal{N}$, due to the dense deployment of BSs. We consider regions instead of individual users because service caching is a relatively long-term decision which cannot be updated very frequently, and region captures statistical information of user task requests. Each BS $n$ has a storage space $C_n$, which can be used to store data (e.g. libraries and databases) associated with specific computing services, and a CPU of maximum frequency $f_n$ (cycles per second), which can be used to process tasks offloaded from end users.

Service is an abstraction of applications that is hosted by the BS and requested by mobile users. Example services include video streaming, social gaming, navigation, augmented reality. Running a particular service requires caching the associated data, such as required libraries and databases, on the BS. We assume that there is a set of $K$ computing services, indexed by $\mathcal{K} = \{1,2,...,K\}$. Each service $k$ requires a storage space $c_k$. For service $k$, we assume that the workload (in terms of the required CPU cycles) of one corresponding task follows an exponential distribution with mean $\mu_k$. Therefore, services are heterogeneous in terms of both required storage and CPU.

The computation demand for service $k$ in time slot $t$ is described by a vector $\d^t_k = (d^t_{k,1},...,d^t_{k,M})$ where $d^t_{k,m}$ is the demand intensity generated by users in region $m$. Specifically, we consider that the task arrival  follows a Poisson process at rate $d^t_{k,m}$. In practice, a demand predictor can estimate the instantaneous demand prior to the beginning of time slot $t$ using some well-studied learning techniques (e.g. auto-regression analysis). Note that this prediction is short-term, only for the immediate next time slot, which is different from the long-term prediction required by an offline algorithm. Many prior studies show that such instantaneous workload can often be predicted with a high accuracy \cite{liu2012renewable}.

\subsection{Service Caching and Task Offloading Decisions}
At the beginning of each time slot $t$, each BS $n$ makes two decisions: service caching and task offloading.

\subsubsection{Service Caching}
Caching service $k$ allows tasks requiring service $k$ to be processed at the network edge, thereby reducing computation latency and improving user quality of experience. However, due to the limited storage space of a BS, not all services can be cached at the same time. Therefore, the BS has to judiciously decide which services to cache. Specifically, BS $n$ decides to cache a subset of services among $\mathcal{K}$. Let $a^t_{n,k} \in \{1, 0\}$ be a binary decision variable to denote whether service $k$ is cached or not on BS $n$ in time slot $t$. The service caching decision of BS $n$ is collected in $\a^t_n=\{a^t_{n,1},a^t_{n,1},\dots,a^t_{n,K}\}$.  Moreover, the service caching decisions are subject to the following capacity constraint
\begin{align}
\sum_{k} a^t_{n,k}c_k \leq C_n, \forall t, \forall n
\end{align}

Let $\mathcal{A}^t_{m,k} \subseteq \mathcal{B}_m$ denote the set of BSs that have service $k$ cached in time slot $t$ and hence can provide the corresponding computing service to region $m$. For analytical simplicity, we assume that demand $d^t_{k,m}$ in region $m$ is evenly distributed among BSs in $\mathcal{A}^t_{m,k}$. Nevertheless, other user-cell association rules (e.g. a user offloads task to the BS with the best channel condition among the ones who can provide the required service) can also be easily incorporated in our model. The demand for service $k$ to BS $n$ can be computed as $\lambda^t_{n,k} = a^t_{n,k}\sum_{m=1}^M\mathbf{1}\{n\in\mathcal{B}_m\}\frac{d^t_{k,m}}{|\mathcal{A}^t_{m,k}|}$.
Note that if there is no BS in $\mathcal{B}_m$ providing service $k$, then all tasks demanding service $k$ will be offloaded to the remote cloud for processing via any nearby BS. Let $\a^t=\{a^t_{n,1},a^t_{n,1},\dots,a^t_{n,K}\}$ denote the service caching decisions of all BSs in time slot $t$.

\subsubsection{Task Offloading}
Among the set of cached services, BS $n$ also has to decide the amount of tasks that are processed locally at the network edge. The remaining tasks will be offloaded to the remote cloud. Let $b^t_n \in [0,1]$ be a continuous decision variable to denote the fraction of service tasks that are processed locally at BS $n$. Hence, the amount of locally processed tasks is $b^t_n \sum_k \lambda^t_{n,k}$. We note that the actual task offloading actions are performed during the time slot when the tasks actually arrive and will depend on the specific task requirements. Nevertheless, the task offloading decisions in terms of the fraction of offloaded tasks can still be planned at a reasonably high granularity at the beginning of each time slot. Let $\b^t=\{b^t_1,b^t_2,\dots,b^t_n\}$ denote the task offloading decisions of all BSs in time slot $t$.

\subsection{Energy Consumption and Computation Delay Cost}
Different service caching and task offloading decisions result in different computation latency performance and incurs different computing energy consumption.

\subsubsection{Energy consumption} The BS dynamically adjusts its CPU speed depending on the task workload. To simplify our analysis, we assume that the BS processes tasks at its maximum CPU speed while choosing the minimum CPU speed when it is idle.  Assuming that the BS consumes a negligible energy under the minimum speed mode, the average computation energy consumption can be expressed as \cite{ren2013coca}:
\begin{align}
E^t_n(\a^t, \b^t) = \gamma_n + \kappa_n b^t_n \sum_{k} \mu_k a^t_{n,k} \lambda^t_{n,k}
\end{align}
where $\gamma_n$ is the static power regardless of the workload as long as BS $n$ is turned on, and $\kappa_n$ is the unit energy consumption when BS $n$ is processing tasks at its maximal speed $f_n$. In the above equation, $b^t_{n} \sum_k \mu_k a^t_{n,k} \lambda^t_{n,k}$ is the expected total number of CPU cycles required to process tasks at BSs.

In addition to computation energy consumption, the BS also incurs energy consumption due to load-independent operations. We denote it by $\tilde{E}_n^t$ which varies over time but can only be observed at the end of each time slot $t$.

\subsubsection{Computation delay cost}
To quantify the overall network performance, we introduce the notion of delay cost capturing the delay-induced revenue loss and/or user dis-satisfaction. The average computation delay for tasks processed by BS $n$ can be computed by modeling the service process as an M/G/1 queue and analyzing its sojourn time (i.e. service time plus waiting time). Since the task arrival of each service type is assumed to follow a Poisson process, the overall task arrival process (without differentiating the specific service types) is also Poisson. Let $\tilde{\lambda}^t_{n,k} = b^t_{n}\lambda^t_{n,k}$ be the task processed at BS $n$ for service $k$ and $\tilde{\lambda}^t_{n} = \sum_k \tilde{\lambda}^t_{n,k}$ be the total workload, which are results of the service caching and task offloading decisions of all BSs. Since there are possibly multiple types of services, the overall service time distribution is a random sampling among a number of exponential distributions. Specifically, the probability of the exponential distribution with mean $\mu_k$ being sampled is $\tilde{\lambda}^t_{n,k}/\tilde{\lambda}^t_n$. Let $s$ be the random variable representing the service time. Its first and second moments can be derived as$\mathbb{E}[s] = \sum_{k} \mu_k \tilde{\lambda}^t_{n,k}/\tilde{\lambda}^t_n$, $
\mathbb{E}[s^2] = \sum_{k} 2\mu^2_k \tilde{\lambda}^t_{n,k}/\tilde{\lambda}^t_n$. According to the Pollaczek-Khinchin formula for M/G/1 queuing system \cite{kleinrock1976queueing}, the expected sojourn time is therefore
\begin{align}
T^t_n(\a^t, \b^t) = \dfrac{1}{f_n}\mathbb{E}[s] + \frac{\tilde{\lambda}^t_n \mathbb{E}[s^2]}{2(f_n- \tilde{\lambda}_n^t \mathbb{E}[s])}
\end{align}
We assume that the remote cloud has ample computing power and hence, the computation delay for tasks offloaded to the cloud, denoted by $h^t$, is mainly due to the transmission delay. Therefore, the total expected computation delay cost for tasks arriving at BS $n$ is
\begin{align}
&D^t_n(\a^t, \b^t) = \tilde{\lambda}^t_n T^t_n(\a^t, \b^t) + (\lambda^t_n - \tilde{\lambda}^t_n) h^t \nonumber\\
&= \lambda^t_n h^t + \sum_k(\mu_k -h^t)\tilde{\lambda}^t_{n,k} + \frac{\sum_{k} \tilde{\lambda}^t_{n,k} \sum_k \mu^2_k \tilde{\lambda}^t_{n,k}}{1 - \sum_k\mu_k \tilde{\lambda}^t_{n,k}}
\end{align}

\subsection{Problem Formulation}
The goal of the network operator is to make joint service caching and task offloading decisions to minimize computation latency while keeping the total computation energy consumption low. The problem is formulated as follows:
\begin{subequations}
\begin{align}
(\textbf{P1})~\min_{\a^t, \b^t, \forall t}&~~\lim_{T\to\infty}\frac{1}{T}\sum_{t=1}^T\bigg[\sum_{n=1}^N D^t_n(\a^t, \b^t)\nonumber\\&~~~~~~~~~~~~~~~~~+h^t(\sum_{m,k}d^t_{k,m}-\sum^{N}_{n=1}\lambda^t_n)\bigg]\\
\text{s.t.}~&\lim_{T\to\infty}\frac{1}{T}\sum_{t=1}^T\sum_{n=1}^N \left(E^t_n(\a^t, \b^t)  + \tilde{E}^t_n \right)\leq Q \label{energy_constraint}\\
~~&\sum_{k}a^t_{n,k}c_k \leq C_n, \forall t, \forall n \label{storage_capacity}\\
~~& E^t_n(\a^t, \b^t)  + \tilde{E}^t_n\leq E^{\max}_n \label{max_energy}\\
~~& D^t_n(\a^t, \b^t)\leq D^{\max}_n\label{max_delay}
\end{align}
\end{subequations}
where $h^t$ is the service delay for tasks whose service data is not cached in the BSs. The first constraint \eqref{energy_constraint} is the long-term energy constraint for the network of BSs, which requires that the long-term average total energy consumption does not exceed an upper limit $Q$. This constraint couples the BS decision making both spatially (i.e. across BSs) and temporally (i.e. across time slots). The second constraint \eqref{storage_capacity} is due to the individual BS's storage space capacity. The third and fourth conditions \eqref{max_energy} and \eqref{max_delay} impose per-slot constraints on the maximum energy consumption and delay, respectively, for each BS.

The first major challenge that impedes the derivation of the optimal solution to the above problem is the lack of future information: optimally solving \textbf{P1} requires complete offline information (distributions of task demands in all time slots) which is difficult to predict in advance, if not impossible. Moreover, \textbf{P1} is a mixed integer nonlinear programming and is very difficult to solve even if the future information is known a priori. These challenges call for an online approach that can efficiently make service caching and offloading decisions on-the-fly without foreseeing the future.

\section{Online Service Caching and Task Offloading}
In this section, we first develop our online algorithm, called OREO (Online seRvice caching for mobile Edge cOmputing), and then show that it is efficient in terms of latency minimization compared to the optimal offline algorithm. Our algorithm is developed under the Lyapunov optimization framework which converts the original long-term optimization problem \textbf{P1} to per-slot optimization problems requiring only current slot information. Our algorithm also enables BSs to decide which services to cache and how much task workload to retain at the edge or to offload to the cloud in a distributed way.

\subsection{Lyapunov-based Online Algorithm}
A major challenge of directly solving \textbf{P1} is that the long-term energy constraint of BSs couples the service caching and task offloading decisions across different time slots. To address this challenge, we leverage the Lyapunov optimization technique and construct a (virtual) energy deficit queue to guide the service caching and task offloading decisions to follow the long-term energy constraint. Specifically, assuming $q(0) = 0$, we construct an energy deficit queue whose dynamics evolves as follows
\begin{align}
q(t+1) = \left[q(t) + \sum_n\left(E^t_n(\a^t, \b^t) + \tilde{E}^t_n \right) - Q\right]^+
\end{align}
where $q(t)$ is the queue backlog in time slot $t$ indicating the deviation of current energy consumption from the energy constraint. The Lyapunov function is defined as $L(q(t)) \triangleq \frac{1}{2}q^2(t)$, representing the ``congestion level" in energy deficit queue. A small value of $L(q(t))$ implies that the queue backlog is small, which means that the virtual queue has strong stability. To keep the energy deficit queue stable, i.e., to enforce the energy consumption constraints by persistently pushing the Lyapunov function towards a lower value, we introduce \emph{one-slot Lyapunov drift}, which is $\Delta(q(t))=\mathbb{E}[L(q(t+1))-L(q(t))|q(t)]$. Then we have
  \begin{align}
  \Delta(q(t))&=\frac{1}{2}\mathbb{E}\left[q^2(t+1)-q^2(t)\mid q(t)\right]\\
  &\stackrel{(\dag)}{\leq}\frac{1}{2}\mathbb{E}\left[(q(t)+\hat{E}^t-Q)^2-q^2(t)\mid q(t)\right]\\
 &=\frac{1}{2}(\hat{E}^t-Q)^2+q(t)\mathbb{E}\left[(\hat{E}^t-Q) \mid q(t)\right]\\
 &\leq B + q(t) \mathbb{E}\left[(\hat{E}^t-Q) \mid q(t)\right]
 \end{align}
where $\hat{E}^t=\sum_{n}(E^t_n(\a^t,\b^t) + \tilde{E}^t_n)$ and $B=\frac{1}{2}\left(\sum_nE_n^{\max}-Q\right)^2$. The inequality ($\dag$) comes from $(q(t)+\hat{E}^t-Q)^2 \geq [(q(t)+\hat{E}^t-Q)^+]^2$.

Under the Lyapunov optimization framework, the underlying objective of our optimal control decision is to minimize a supremum bound on the following \emph{drift-plus-cost} expression in each time slot:
\begin{align}\label{drift_plus_cost}
  &\Delta(q(t))+V\cdot\mathbb{E}\left[\hat{D}^t(\a^t,\b^t) \mid q(t) \right]\\
  &\leq B + q(t) \mathbb{E}\left[(\hat{E}^t-Q) \mid q(t)\right]+V\cdot\mathbb{E}\left[\hat{D}^t(\a^t,\b^t) \mid q(t) \right]\nonumber
\end{align}
where $\hat{D}^t=\sum_{n=1}^N D^t_n(\a^t, \b^t)+h^t(\sum_{m,k}d^t_{k,m}-\sum^{N}_{n=1}\lambda^t_n)$.

Our proposed algorithm OREO minimizes the right hand side of \eqref{drift_plus_cost} (see Algorithm 1). The network determines the service caching and task offloading strategies in each time slot by solving the optimization problem \textbf{P2} as follows
\begin{subequations}\label{P2}
\begin{align}
(\textbf{P2})~~\min_{\a^t, \b^t}~& \left(V\cdot \hat{D}^t(\a^t, \b^t) + q(t)\cdot \hat{E}^t(\a^t, \b^t)\right)\\
\text{s.t.}~~&\eqref{storage_capacity}, \eqref{max_energy}, \eqref{max_delay}
\end{align}
\end{subequations}

\begin{algorithm}
\caption{OREO algorithm}
\begin{algorithmic}[1]
\Statex \textbf{Input}: $q(0) \gets 0$, $\mu_k$, $c_k$, $C_n$, $E_n^{max}$,$D_n^{max}$;
\Statex \textbf{Output}: service caching decision $\{\a^1, \a^2, \dots, \a^T\}$, offloading decisions $\{\b^1, \b^2, \dots, \b^T\}$;
\For{$t=0~\text{to}~T-1$}
\State Predict service demand $d^t_{m,n}$;
\State Observe $h^t$, $\tilde{E}^t_n$, $R^t_n$;
\State Choose $\a$, $\b$ by solving \textbf{P2}:
\State $q(t+1)=[q(t) + \sum_n(E^t_n(\a^t, \b^t) + \tilde{E}^t_n ) - Q]^+$;
\EndFor
\end{algorithmic}\label{OREO}
\end{algorithm}

The positive parameter $V$ is used to adjust the tradeoff between computation latency minimization and the energy consumption minimization of BSs. Notice that solving \textbf{P2} requires only currently available information as input. By considering the additional term $q(t)\cdot \sum_n E^t_n(\a^t, \b^t)$, the network takes into account the energy deficit of BSs during current-slot service caching and task offloading. As a consequence, when $q(t)$ is larger, minimizing the energy deficit is more critical. Thus, our algorithm works following the philosophy of ``if violate the energy constraint, then use less energy'', and the energy deficit queue maintained without foreseeing the future guides the BSs towards meeting the energy constraint, thereby enabling online decision making. Now, to complete the algorithm, it remains to solve the optimization problem \textbf{P2}, which will be discussed in the next subsection.

\subsection{Distributed Optimization for \textbf{P2}}
In this subsection, we focus on solving \textbf{P2} which is a joint optimization problem aiming to find the optimal service caching and task offloading decisions for each time slot $t$. Since service caching decisions are binary and task offloading decisions are continuous, \textbf{P2} is a mixed-integer nonlinear programming. While there exist various centralized techniques (such as Generalized Benders Decomposition \cite{geoffrion1972generalized}) to solve it, these methods are usually computationally prohibitive and distributed solutions are much desired. In this paper, we present a distributed algorithm based on a variation of Gibbs sampling \cite{christian1999monte}, which determines the optimal decision pair $(\a^t, \b^t)$ in an iterative manner at the beginning of a time slot.

The algorithm works as follows. In each iteration, a randomly selected BS $n$ virtually changes its current service caching decision $\a_n^t$ to $\tilde{\a}_n^t$ (Line \ref{randm_pick}) and then the optimal offloading scheme $\b^t$ is derived by solving \eqref{P2}. However, when deriving $\b^t$, only neighboring BSs (i.e. BSs that have overlapping service areas with BS $n$) need to make the update since $\a_n^t$ affects the traffic distribution only among the neighborhood of BS $n$. Afterwards, the new delay cost $\tilde{f}$ restricted to the neighborhood of BS $n$ is obtained, and the service caching action of BS $n$ is updated to the new action $\tilde{\a}_n^t$ with probability $\eta$ and keeps unchanged (i.e. $\a^t_n$) with probability $1-\eta$ depending on the delay cost difference $\tilde{f} - f$ (Lines \ref{eta} and \ref{evolve}). Therefore, changing service caching decision is more likely to occur if the new action $\tilde{\a}_n^t$ results in a lower delay cost. At the end of the iteration, BS $n$ broadcast its current service caching decision to its neighboring BSs.

\textbf{Remark}: It is known that always choosing a better decision in combinatorial optimization can easily lead to a local optimality. To avoid being trapped in a local optimum, the proposed algorithm explores a new decision with a certain probability even though it may be worse than the current decision (i.e. $\tilde{f} > f$). The parameter $\tau>0$, referred as the smooth parameter, is used to control exploration versus exploitation (i.e. the degree of randomness). When $\tau$ is small, the algorithm tends to keep a new decision with larger probability if it is better than the current decision.  However, in this case, it takes more iterations to identify the global optimum since the algorithm may be stuck in a local optimum for a long time before exploring other solutions that lead to more efficient decisions. When $\tau\to+\infty$, the algorithm tries to explore all possible solutions from time to time without convergence.

\textbf{Remark}: The convergence rate of this algorithm can be further improved by letting multiple BSs evolve their service caching decisions simultaneously in each iteration, provided that they are sufficiently apart. Specifically, if a set of BSs do not have common neighboring BSs between any pair, then their service caching decisions do not affect each other and hence, they are allowed to evolve simultaneously.

\begin{algorithm}[b]
\caption{Distributed algorithm for OREO}
\begin{algorithmic}[1]
\Statex \textbf{Input}: service cache decision $\a^t\gets\bm{0}$; task offloading decision $\b^t\gets\bm{0}$; objective value $f\gets+\infty$;
\State Predict the service demand $\d^t_k$;
\State Randomly pick a BS $n\in\mathcal{N}$ and select service cache decision $\tilde{\a}_n\in\Phi$; \label{randm_pick}
\If{$\tilde{\a}^t_n$ is \emph{feasible}}
\State $\tilde{\a}^t\gets\{\a^t_{-n},\tilde{\a}^t_n\}$;
\State Observe computing demand $\lambda^t_{n,k}, \forall n, \forall k$;
\State Obtain $\tilde{\b}^t,\forall k$ by minimizing \textbf{P2}:
\Statex \qquad$\min\limits_{\tilde{\b}^t}~V\hat{D}^t(\tilde{\a}^t,\tilde{\b}^t) + q(t)\cdot \hat{E}^t(\tilde{\a}^t,\tilde{\b}^t)$
\State $\eta\gets \dfrac{1}{1+e^{(\tilde{f}-f)/\tau}}$ \label{eta}
\State With probability $\eta$, BS $n$ sets $\a^t_n\gets\tilde{\a}^t_n, \b^t\gets\tilde{\b}^t, f\gets\tilde{f}$; with probability $(1-\eta)$, BS $n$ keeps $\a^t_n$ unchanged \label{evolve}
\State Broadcast $\a^t_n$ to its neighboring BSs
\EndIf
\State Return $\a_n^t, \b^t$ if the stopping criterion is satisfied, otherwise, go to Line \ref{randm_pick}
\end{algorithmic}\label{dis_alg_p2}
\end{algorithm}

Next, we formally prove the convergence of our algorithm.

\begin{theorem} \label{converge_optimality}
  As $\tau>0$ decreases, the algorithm converges with a higher probability to the global optimal solution of \textbf{P2}. When $\tau \to 0$, the algorithm converges to the globally optimal solution with probability 1.
\end{theorem}

\begin{proof}
See Appendix \ref{proof_converge_optimality}
\end{proof}

\subsection{Performance Analysis}
This subsection presents the performance analysis of OREO using the Lyapunov optimization technique.
\begin{theorem}\label{OREO_performance_guarantee}
By applying OREO, the time-average system delay satisfies:
   \begin{equation*}\label{system_delay_bound}
        \lim_{T\rightarrow\infty}\dfrac{1}{T}\sum_{t=0}^{T-1}\mathbb{E}\left[\hat{D}^t(\a^t,\b^t)\right]<\hat{D}^{opt}+\dfrac{B}{V}
   \end{equation*}
and the time-average energy consumption of BSs satisfies:
\begin{equation*}\label{energy_defit_bound}
\lim_{T\rightarrow\infty}\dfrac{1}{T}\sum_{t=0}^{T-1}\mathbb{E}\left[\hat{E}^t(\a^t,\b^t)\right]\leq\dfrac{B}{\epsilon}+\dfrac{V}{\epsilon}(\hat{D}^{max}-\hat{D}^{opt})+Q
\end{equation*}
where $\hat{D}^{opt}=\lim\limits_{T\to\infty}\dfrac{1}{T}\sum\limits_{t=0}^{T-1} \sum\limits_{n=1}^{N}\mathbb{E}\left\{D_n^t(\a^{opt,t},\b^{opt,t})\right\}$ is the optimal system delay to $\textbf{P2}$, $\hat{D}^{max}$ is the largest system delay, and $\epsilon>0$ is a constant which represents the long-term energy surplus achieved by some stationary control strategy.
\end{theorem}

\begin{proof}
	See Appendix \ref{proof_OREO_performance_guarantee}.
\end{proof}	

The above theorem demonstrates an $[O(1/V ), O(V )]$ delay-energy tradeoff. OREO asymptotically achieves the optimal performance of the offline problem \textbf{P1} by letting $V\to\infty$. However, the optimal performance of \textbf{P1} is achieved at the price of a higher energy consumption, as a larger energy deficit queue is required to stabilize the system and hence convergence is postponed. This also implies that the time-average energy consumption grows linearly with $V$.

\section{Simulation}
In this section, we carry out simulations to evaluate the performance of OREO. We simulate a 500m$\times$500m area served by 9 BSs regularly deployed on a grid network. The serving radius for each BS is set as 150m. The whole area is divided into 25 regions and the demand for service $k$, $d^t_{m,k}~\forall k $, in region $m$ during slot $t$ follows a uniform distribution $d^t_{m,k}\in[0,12]$. The actual service demand in region $m$ is formulated as a Poisson process with arrival rate $d^t_{m,k}$. Other main parameters are collected in Table \ref{para_set}.

\begin{table}
	\renewcommand\arraystretch{1}
	\centering
	\caption{Simulation setup: system parameters}
	\begin{tabular}{l|c}
		\hline
		
		Parameters & Value\\
		\hline
		\hline
		Total service types, $K$ & 10 \\
		BS service rate $f_n$   & 10 GHz \\
		BS communication distance & 130m\\
		CPU cycles requirement of service $k$,  $\mu_k$  &  $[0.1, 0.5]$ GHz/task \\
		BS storage space $C_n$ & 200 GB\\
		Storage space requirement of service $k$ $c_k$ & $[20,100]$ GB \\
		Unit energy consumption, $\kappa_n$  & 1 kWh\\
		Computation delay for tasks offloaded to cloud, $h$& $[2,4] $ sec/task\\
		Smooth factor, $\tau$  & $10^{-2}$\\
		Energy for cooling and communication traffic, $\tilde{E}^t_n$ & $[0,3]$ kWh\\
		\hline
	\end{tabular}
	\label{para_set}
\end{table}

We compare OREO with three benchmarks. \textbf{Non-cooperative service caching}: BSs cache services with the largest demand in the serving region. Each BS works independently without mutual communication and the long-term energy consumption constraint is ignored. \textbf{Centralized delay-optimal service caching}: A centralized service caching decision is found for all BSs to minimize the system delay. The decision is made regardless of the long-term energy consumption constraint. \textbf{Myopic service caching}: We impose a hard energy consumption constraint in each time slot and minimize the system delay. This method can also satisfy the long-term energy constraint without requiring future information. However, it is less adaptive and purely myopic.

\subsection{Performance comparison}
Fig. \ref{PC_delay} and Fig. \ref{PC_eng} show the time average system delay and energy consumption, respectively. It shows that OREO achieves near-to-optimal delay performance while closely following the long-term energy constraint. The centralized optimal scheme achieves the lowest system delay as expected. However, it is achieved at a cost of large energy consumption. By contrast, OREO slightly sacrifices the delay performance to satisfy the energy consumption constraint. For the myopic service caching, because a hard energy constraint is imposed in every time slot, the long-term energy consumption constraint is also satisfied. However, because very little energy consumption is incurred due to very low task demand in some time slots, the time average energy consumption can be far below the long-term constraint, resulting in inefficient energy usage. In the non-cooperative case, BSs make decisions individually based on their predicted demand. This strategy neglects the interdependence among BSs and results in inferior performance in both system delay minimization and energy saving.

\begin{figure*}
	\begin{minipage}[t]{0.33\linewidth}
		\centering
		\includegraphics[width=0.95\textwidth]{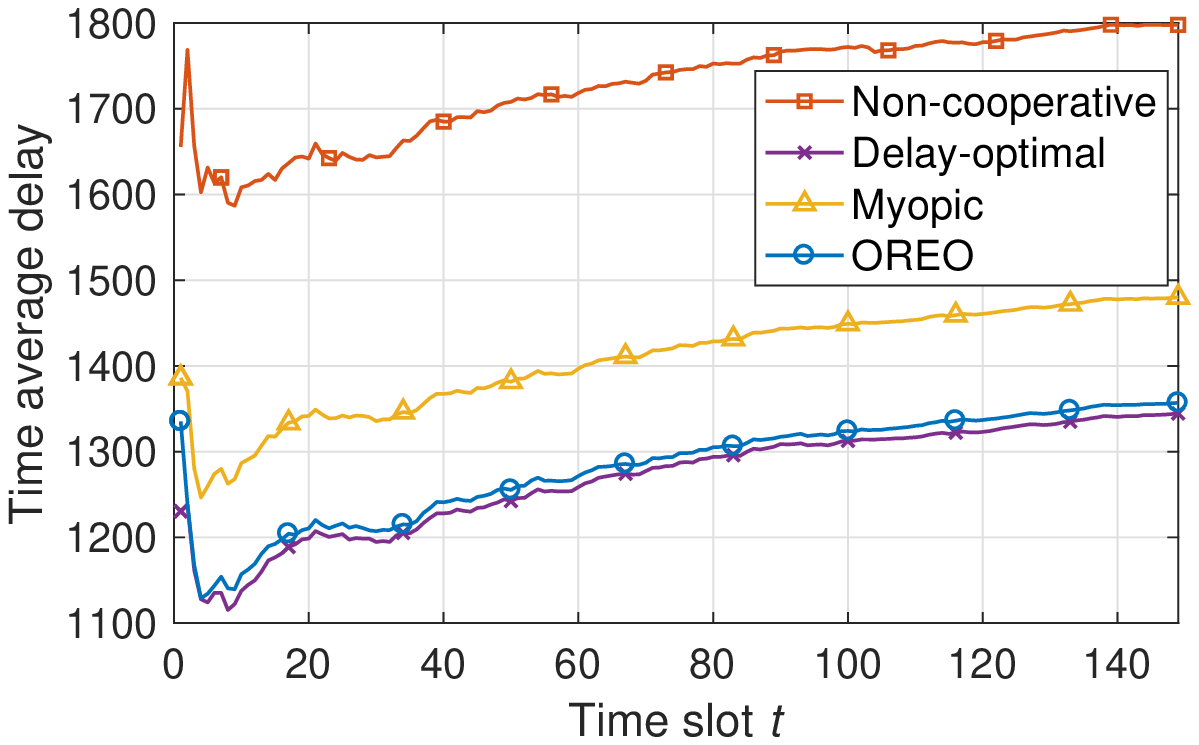}
		\caption{Time average delay}\label{PC_delay}
	\end{minipage}%
	\begin{minipage}[t]{0.33\linewidth}
		\centering
		\includegraphics[width=0.95\textwidth]{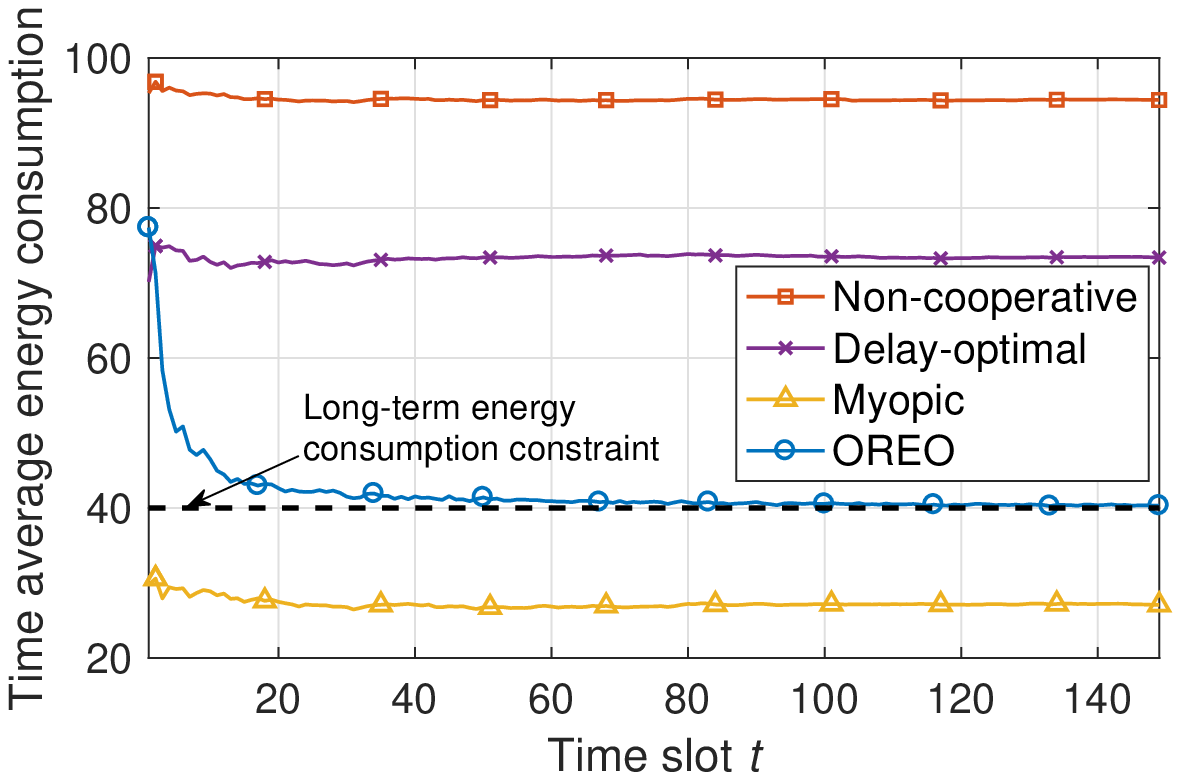}
		\caption{Time average energy consumption}\label{PC_eng}
	\end{minipage}%
	\begin{minipage}[t]{0.33\linewidth}
		\centering
		\includegraphics[width=0.95\textwidth]{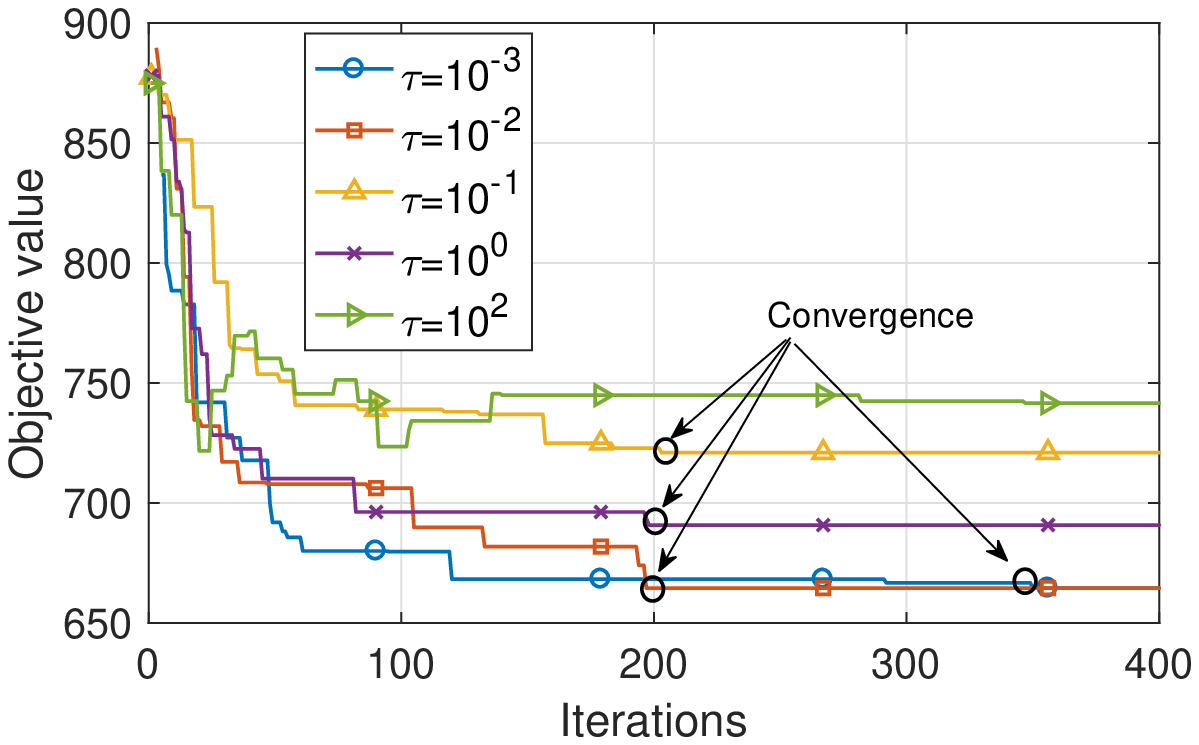}
		\caption{Convergence of distributed algorithm}\label{tau_evo}
	\end{minipage}%
\end{figure*}

\begin{figure*}
	\begin{minipage}[t]{0.33\linewidth}
		\centering
		\includegraphics[width=0.95\textwidth]{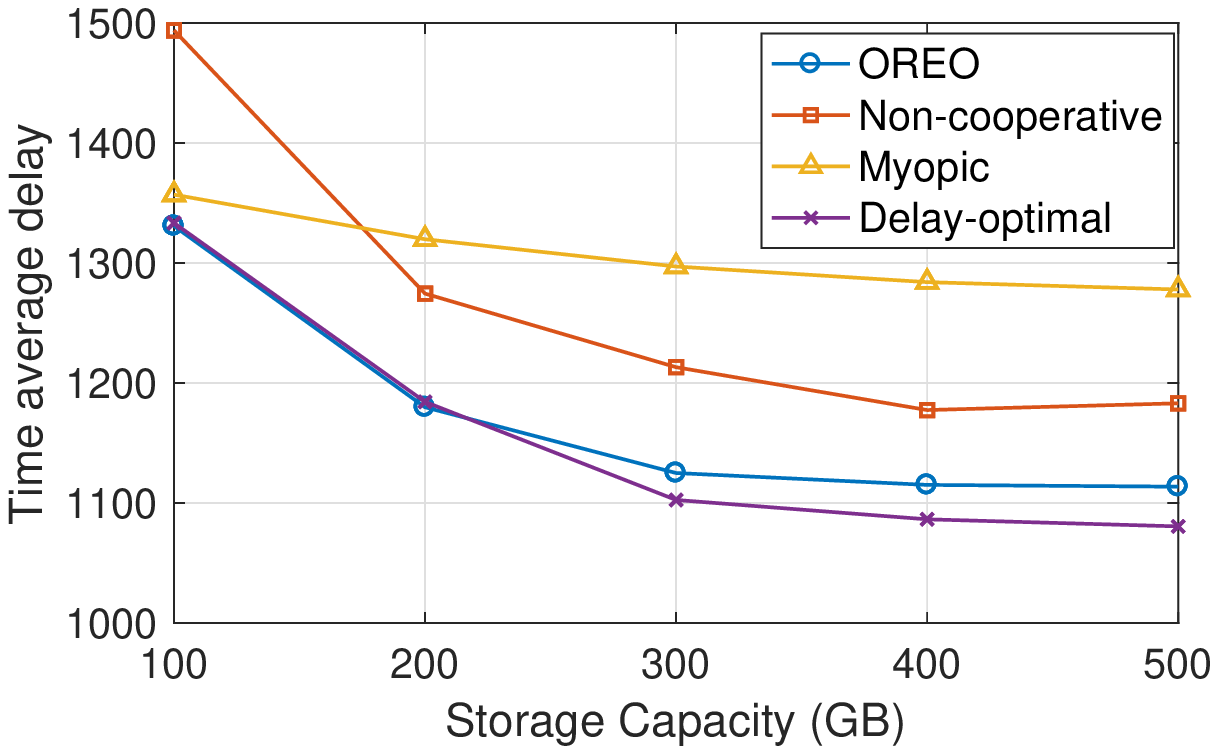}
		\caption{Time average delay with different $C$}\label{vary_C_delay}
	\end{minipage}%
	\begin{minipage}[t]{0.33\linewidth}
		\centering
		\includegraphics[width=0.95\textwidth]{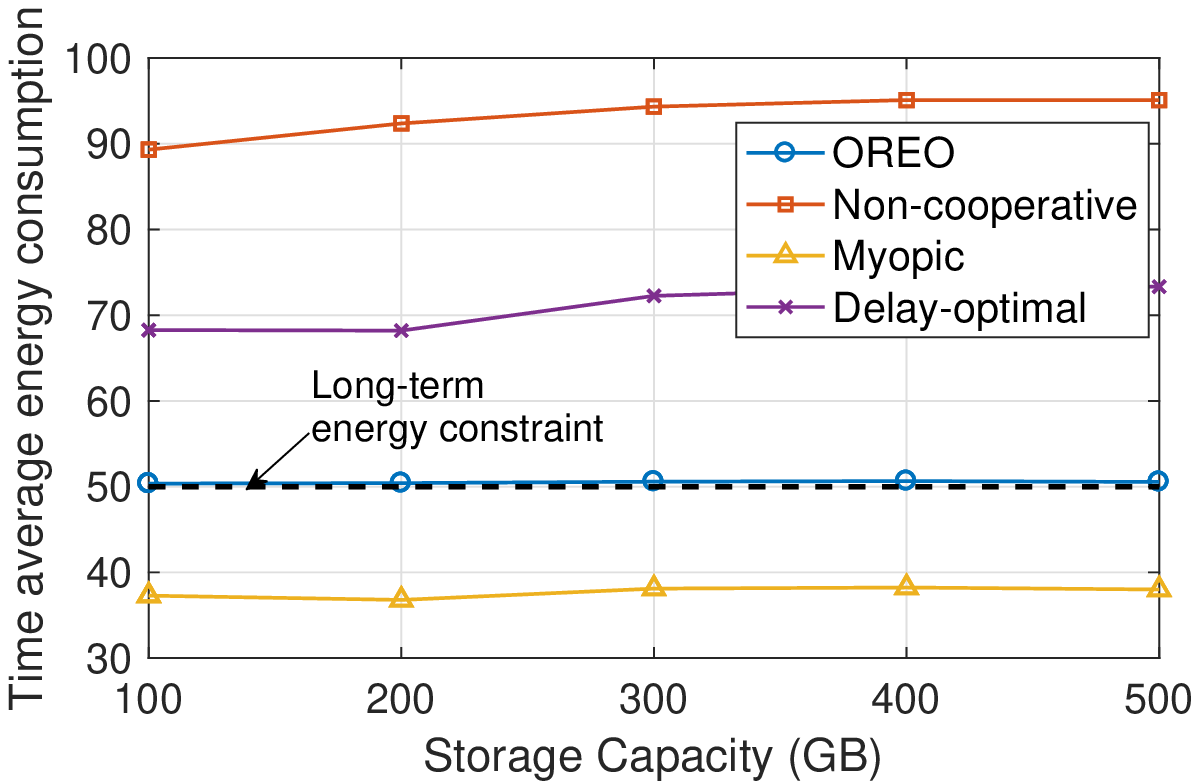}
		\caption{Time average energy with different $C$}\label{vary_C_eng}
	\end{minipage}%
	\begin{minipage}[t]{0.33\linewidth}
		\centering
		\includegraphics[width=1\textwidth]{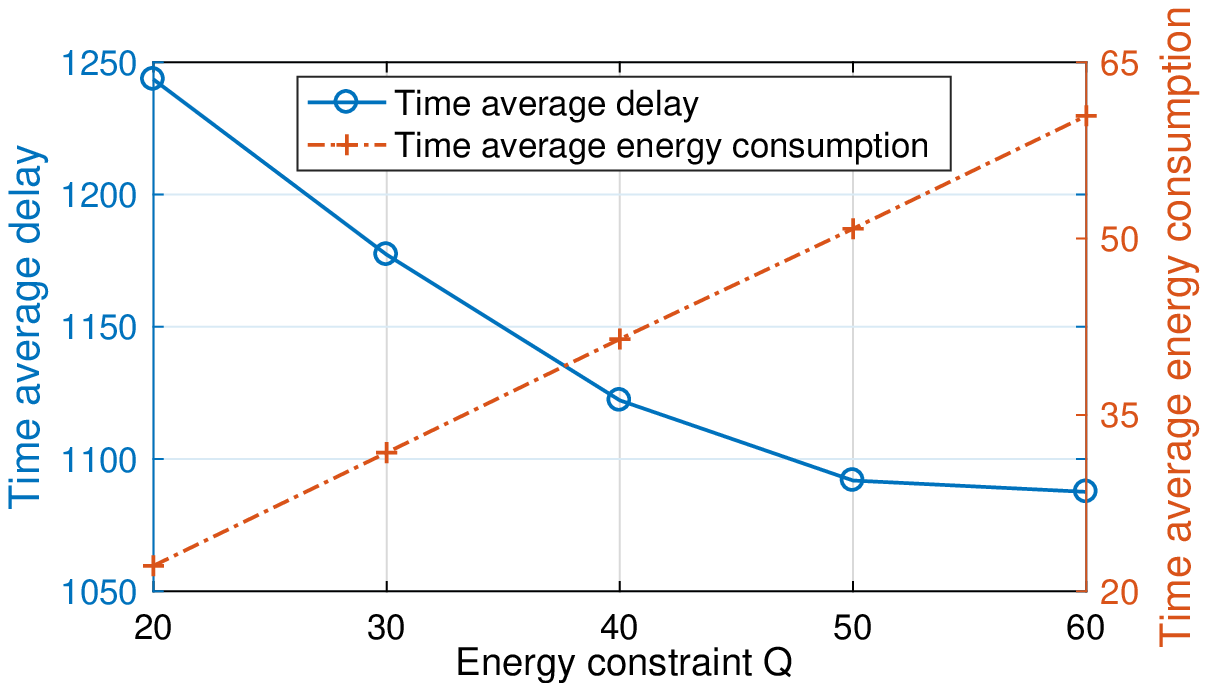}
		\caption{Time average delay with different $Q$}\label{Q_delay_eng}
	\end{minipage}%
\end{figure*}

\begin{figure*}
	\begin{minipage}[t]{0.33\linewidth}
		\centering
		\includegraphics[width=1\textwidth]{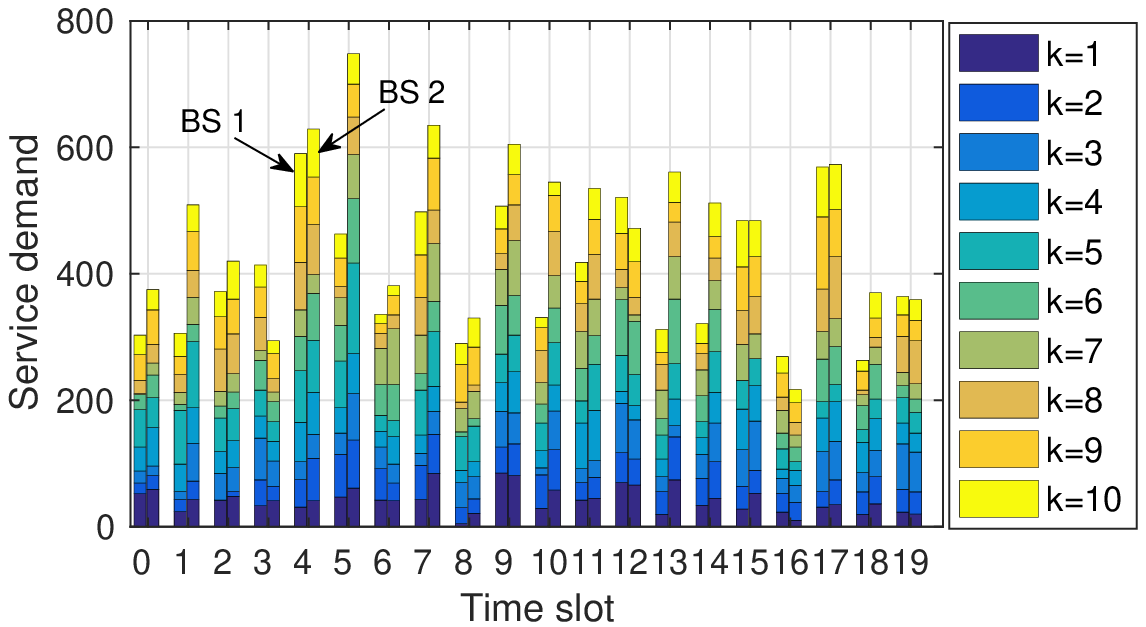}
		\caption{Predicted service demand}\label{demand_t}
	\end{minipage}%
	\begin{minipage}[t]{0.33\linewidth}
		\centering
		\includegraphics[width=1\textwidth]{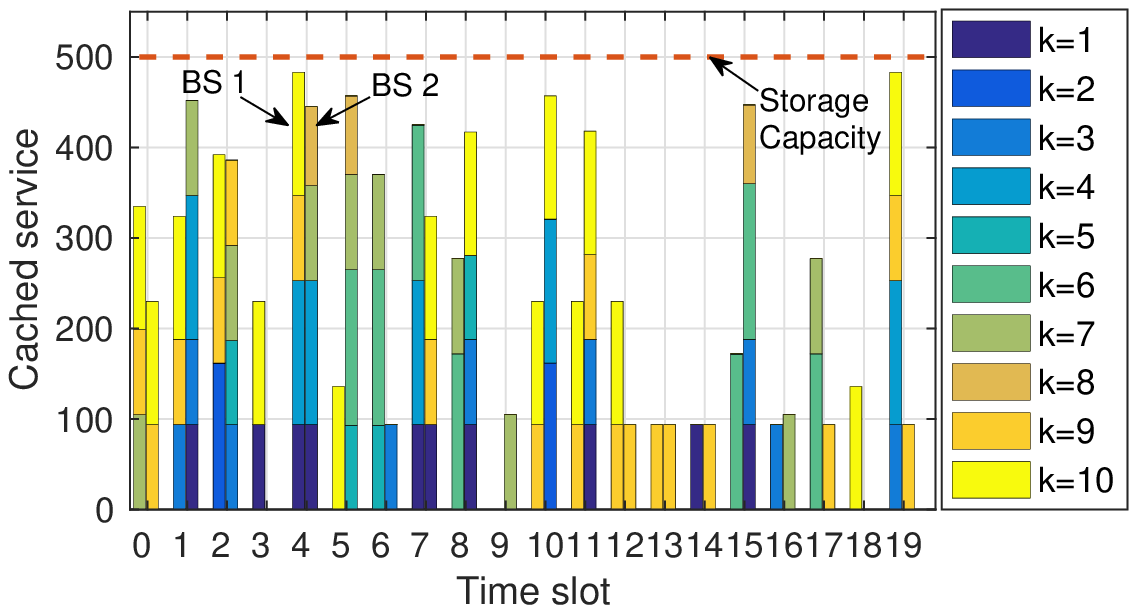}
		\caption{Cached service}\label{cached_serv}
	\end{minipage}%
	\begin{minipage}[t]{0.33\linewidth}
		\centering
		\includegraphics[width=0.98\textwidth]{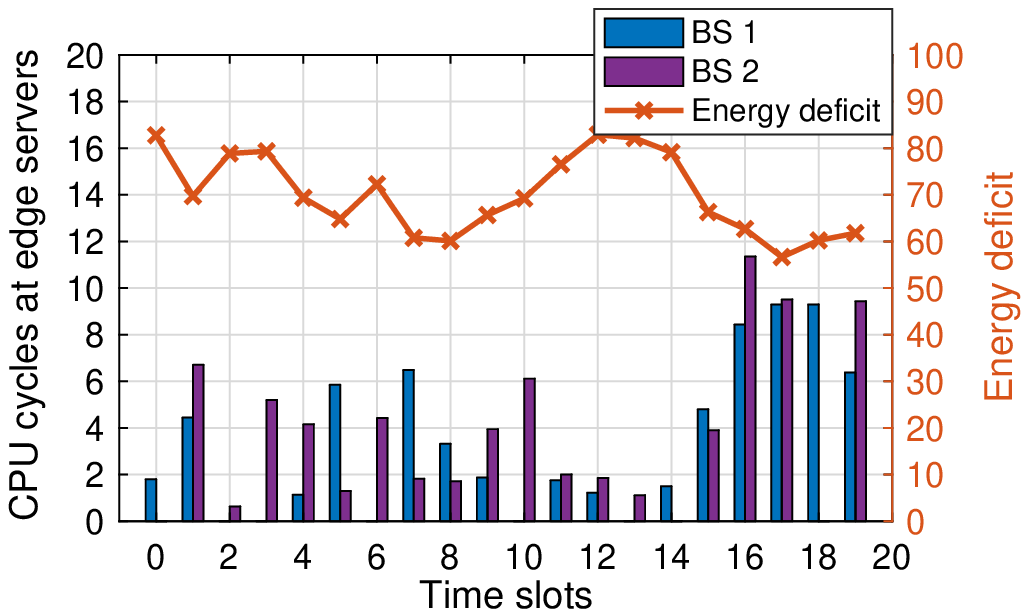}
		\caption{Energy deficit and workload at edge}\label{cycle_deficit}
	\end{minipage}%

\end{figure*}

\subsection{Convergence of the distributed algorithm}
Fig. \ref{tau_evo} shows the convergence process when running the distributed algorithm. We see that when $\tau=10^{-3}$, the algorithm converges quickly to superior decisions. However, it stays in several local optimal solutions for a while before identifying the global optimum. By increasing $\tau$ to $10^{-2}$, the algorithm can find the global optimum with much fewer iterations. However, keeping increasing the parameter $\tau$ impedes the identification of global optimum and results in the convergence to inferior solutions.

\subsection{Impact of storage capacity}
Fig. \ref{vary_C_delay} and \ref{vary_C_eng} show the impact of storage capacity on the converged system delay and energy consumption after 150-slot simulation. It can be observed that the system delay decreases as the storage capacity increases since a larger number of services can be cached at BSs. Moreover, when the storage capacity is small, the system delay achieved by OREO is identical to the centralized delay-optimal scheme. This is due to the fact that few computing tasks are processed at BSs (most of them have to be offloaded to the cloud) and the energy deficit queue $q(t)$ is zero in most time slots and hence, the problem \textbf{P2} degenerates to the delay-optimal format. As the storage capacity increases, the system delay of OREO deviates from that of the delay-optimal scheme as a result of meeting the energy consumption constraint. Fig. \ref{vary_C_eng} shows that the long-term energy consumption of OREO closely follows the energy consumption constraint for all levels of storage capacity, while other three benchmarks either overuse or underuse the predetermined energy budget.

\subsection{Impact of energy constraint}
Fig. \ref{Q_delay_eng} presents the converged time average system delay and energy consumption of OREO under different energy consumption constraints $Q$. It is straightforward to see that the OREO converges to a lower system delay with a larger energy consumption constraint since more tasks are allowed to be processed at BSs. However, the performance gain by increasing the energy constraints become modest when the constraint $Q$ is large. We also see that OREO successfully converges to the predetermined energy consumption constraint.

\subsection{Impact of demand patterns}

Fig. \ref{demand_t} shows the predicted demands for different services of two geographical adjacent base stations, BS 1 and BS 2, in 20 time slots. Since BS 1 and BS 2 have overlapping regions, their demand patterns are correlated. Fig. \ref{cached_serv} shows the corresponding service caching decision across these 20 time slots, where the length of color bar denotes the occupied storage space of services. As can be seen, even when the two BSs have similar demand patterns, their service caching decisions can be dramatically different (see slot 2 for an example). Such cooperation indeed helps to accommodate more types of services and hence more computation tasks, thereby improving the overall system efficiency. Since service caching decisions alone do not determine the system performance and its evolution, we show in Fig. \ref{cycle_deficit} the total number of CPU cycles used to process computation tasks at these two BSs, which is a reflection of the joint service caching and offloading decisions as well as the demand. As can be observed, when the energy deficit is small, the BSs tend to keep more workload locally to minimize the computation delay (e.g. slot 16-18); when the energy deficit is large, BSs process less workload at the edge to reduce the energy consumption. In this way, the energy consumption can be pushed to satisfy the pre-determined constraint.

\section{Conclusion}
In this paper, we studied joint service caching and task offloading for MEC-enabled dense cellular networks. We proposed an efficient online and decentralized algorithm that tailors service caching decisions to both temporal and spatial service popularity patterns. The proposed algorithm is easy to implement while providing provable performance guarantee. There are a few limitations in the current model that demand future research effort. First, user-cell association (load dispatching) decisions can be incorporated in the joint optimization framework. Second, task workload can be further balanced by allowing workload transfer among peer BSs.

\appendix
\subsection{Proof of Theorem \ref{converge_optimality}} \label{proof_converge_optimality}
\begin{proof}
	Let $\Phi=\{\bm{\phi}_1, \bm{\phi}_2, \dots, \bm{\phi}_L\}$ be the action space of service cache decision $\a^t_n$, where $L$ is given by the Bell number $\sum_{k=1}^K {K \choose k}$. At an arbitrary time slot, BS $n$ chooses a service caching decision $\a^t_n\in\Phi$.
	
	For notational convenience, we drop the time index $t$ and denote the BSs' service caching decision by $\a$. Following the iterations in Algorithm \ref{dis_alg_p2}, $\a$ evolves as a N-dimensional Markov chain in which the $i$-th dimension corresponds to BS $i$'s service caching decision. For the ease of presentation, we begin with a 2-BS case and denote the state of the Markov chain as $S_{\a_1,\a_2}$, where $\a_i\in\Phi, i=1,2$. Since only one BS is selected to explore a new service caching decision at each iteration with equal probability among all BSs, we have
	\begin{flalign}\label{trans_prob}
	\text{Pr}(S_{\a^\prime_1,\a^\prime_2} | S_{\a_1,\a_2})=&&
	\end{flalign}
	\begin{equation*}
	\left\{
	\begin{split}
	&\frac{e^{-f(S_{\a^\prime_1,\a^\prime_2})/\tau}}{2L(e^{-f(S_{\a^\prime_1,\a^\prime_2})/\tau}+e^{-f(S_{\a_1,\a_2})/\tau})}, \a^\prime_1= \a_1 ~ \text{or} ~ \a^\prime_2= \a_2 \\
	&0, \qquad\qquad\qquad\qquad\qquad\qquad\qquad \text{otherwise}
	\end{split}	
	\right.
	\end{equation*}
	where $f(S_{\a_1,\a_2})$ is the objective value in \textbf{P2} given $S_{\a_1,\a_2}$.
	
	We then derive the stationary distribution $\text{Pr}^*$ for each state and examine the balanced equation as follows
	\begin{align}\label{balance_equation}
	\sum_{l=2}^{L}\text{Pr}^*&(S_{\phi_1,\phi_1})\times\text{Pr}(S_{\phi_1,\phi_l}|S_{\phi_1,\phi_1})\\
	&=\sum_{l=2}^{L}\text{Pr}^*(S_{\phi_1,\phi_l})\times\text{Pr}(S_{\phi_1,\phi_1}|S_{\phi_1,\phi_l})\nonumber
	\end{align}
	
	By substituting \eqref{balance_equation} with \eqref{trans_prob}, we have
	\begin{align}\label{balance_equation_2}
	&\sum_{l=2}^{L} \text{Pr}^*(S_{\phi_1,\phi_1}) \times \frac{e^{-f(S_{\phi_1,\phi_l})/\tau}}{2L(e^{-f(S_{\phi_1,\phi_1})/\tau}+e^{-f(S_{\phi_1,\phi_l})/\tau})}\\
	&=\sum_{l=2}^{L} \text{Pr}^*(S_{\phi_1,\phi_l}) \times \frac{e^{-f(S_{\phi_1,\phi_1})/\tau}}{2L(e^{-f(S_{\phi_1,\phi_1})/\tau}+e^{-f(S_{\phi_1,\phi_l})/\tau})}\nonumber
	\end{align}
	
	Observing the symmetry of equation \eqref{balance_equation_2}, we note that the set of equations in \eqref{balance_equation_2} are balanced if for arbitrary state $\tilde{S}$ in the strategy space $\Omega$, the stationary distribution is $\text{Pr}^*(\tilde{S})=\mathcal{K}e^{-f(\tilde{S})/\tau}$, where $\mathcal{K}$ is a constant. By applying the probability conservation law, we obtain the stationary distribution for the Markov chain as
	\begin{align}\label{stationary_dist}
	\text{Pr}^*(\tilde{S})=\frac{e^{-f(\tilde{S})/\tau}}{\sum_{S_i\in\Omega}e^{-f(\tilde{S_i})/\tau}}
	\end{align}
	for arbitrary state $\tilde{S}\in\Omega$. In addition, we observe that the Markov chain is irreducible and aperiodic. Therefore, the stationary distribution given in \eqref{stationary_dist} is valid and unique.
	
	Let $S^*$ be the optimal state which yields the minimum value in \textbf{P2}, i.e., $S^*=\arg\max_{S_i\in\Omega}f(S_i)$. From \eqref{stationary_dist}, we have $\lim_{\tau \to 0}\text{Pr}^*(S^*)=1$ which substantiates that the algorithm converges to the optimal state in probability. Finally, the analogous analysis can be straightforwardly extended to an N-dimensional Markov chain, thereby completing the proof.
\end{proof}

\subsection{Proof of Theorem \ref{OREO_performance_guarantee}} \label{proof_OREO_performance_guarantee}
\begin{proof}
To prove the performance guarantee, we first introduce the following Lemma.

\begin{lemma}\label{stationary_policy}
	For any $\delta>0$, there exists a stationary and randomized policy $\Pi$ for \textbf{P2}, which decides $\a^{\Pi,t}, \b^{\Pi,t}$ independent of the current queue backlogs $q(t)$, such that the following inequalities are satisfied: $\mathbb{E}\left[\hat{E}^t(\a^{\Pi,t},\b^{\Pi,t})-Q\right]\leq \delta$.
\end{lemma}
\begin{proof}
	The proof can be obtained by Theorem 4.5 in \cite{neely2010stochastic}, which is omitted for brevity.
\end{proof}

Recall that the OREO seeks to choose strategies that minimizes \textbf{P2} among feasible decisions including the policy in Lemma \ref{stationary_policy} in each time slot. By plugging Lemma \ref{stationary_policy} into the \emph{drift-plus-cost} inequality \eqref{drift_plus_cost}, we obtain
\begin{align}
&\Delta(q(t))+V\mathbb{E}\left[\hat{D}^t(\a^{\Pi,t},\b^{\Pi,t}) \mid q(t) \right]\nonumber\\
&\leq B + q(t) \mathbb{E}\left[(\hat{E}^t(\a^{\Pi,t},\b^{\Pi,t})-Q) \mid q(t)\right]\nonumber\\
&~~~~~~~~~~~+V\mathbb{E}\left[\hat{D}^t(\a^{\Pi,t},\b^{\Pi,t}) \mid q(t) \right]\\
&\stackrel{(\ddag)}{\leq}B + \delta q(t)+V(\hat{D}^{opt}+\delta)\nonumber
\end{align}

The inequality $(\ddag)$ is because that the policy $\Pi$ is independent of the energy deficit queue. By letting $\delta$ go to zero, summing the inequality over $t\in\{0,1,\dots,T-1\}$ and then dividing the result by $T$, we have:
\begin{align}
&\dfrac{1}{T} \mathbb{E} \left[L(q(t))-L(q(0))\right] + \dfrac{V}{T} \sum_{t=0}^{T-1} \mathbb{E} \left[  \hat{D}^t(\a^{\Pi,t},\b^{\Pi,t})\right] \nonumber\\
&\leq B+V\hat{D}^{opt}
\end{align}

Rearranging the terms and considering the fact that $L(q(t))\geq 0$ and $L(q(0))=0$ yields the time average system delay bound.

To obtain the energy consumption bound, we make following assumption: there are values $\epsilon>0$ and $\Psi(\epsilon)$ and an policy $\a^{\Gamma,t}, \b^{\Gamma,t}$ that satisfies:
\begin{align*}
\mathbb{E}\left[\hat{D}(\a^{\Gamma,t}, \b^{\Gamma,t})\right]=\Psi(\epsilon), ~~
\mathbb{E}\left[\hat{E}^t(\a^{\Gamma,t},\b^{\Gamma,t})-Q \right]\leq -\epsilon
\end{align*}

Plugging above into inequality \eqref{drift_plus_cost}
\begin{equation*}
\Delta(q(t))+V\mathbb{E}\left[\hat{D}^t(\a^{\Gamma,t},\b^{\Gamma,t})\right]\leq B+V\Psi(\epsilon)-\epsilon q_(t)
\end{equation*}
Summing the above over $t\in\{0,1,\dots,T-1\}$ and rearranging terms as usual yields:
\begin{equation*}
\begin{split}
&\dfrac{1}{T} \sum_{t=0}^{T-1} \mathbb{E}[q(t)]\leq \dfrac{B+V(\Psi(\epsilon)-\dfrac{1}{T}\sum\limits_{t=0}^{T-1} \mathbb{E}\left[\hat{D}^t(\a^{\Gamma,t},\b^{\Gamma,t})\right]}{\epsilon}\\
&\leq \dfrac{B}{\epsilon}+\dfrac{V}{\epsilon}(\hat{D}^{\max}-\hat{D}^{opt})
\end{split}
\end{equation*}
Considering  $\sum\limits_{t=0}^{T-1}\mathbb{E}[q(t)] \geq \sum\limits_{t=0}^{T-1} \mathbb{E} \left[ \hat{E}(\a^{\Gamma,t}, \b^{\Gamma,t})-Q\right]$ yields the energy consumption bound.
\end{proof}


%

%
%
%
%

\ifCLASSOPTIONcaptionsoff
  \newpage
\fi



%
\bibliographystyle{IEEEtran}
\bibliography{refs}

%

%
%
%




\end{document}